\documentclass[runningheads]{llncs}

\usepackage[hidelinks]{hyperref}
\usepackage{amsmath,amssymb,amsthm,mathtools,stmaryrd}

\usepackage{nicefrac}

\usepackage{xcolor}
\usepackage{tikz}
  \usetikzlibrary{arrows.meta}
  \usetikzlibrary{calc}
  \usetikzlibrary{matrix}
\usepackage{cleveref}
\usepackage{proof}

\makeatletter

\theoremstyle{definition}
\newtheorem{defn}{Definition}

\theoremstyle{plain}

\newtheorem{cor}  [defn]{Corollary}
\newtheorem{lem}  [defn]{Lemma}
\newtheorem{prop} [defn]{Proposition}
\newtheorem{thm}  [defn]{Theorem}
\newtheorem{lemmaAppendix}{Lemma}

\newcommand{\ie}{\textit{i.e.}}
\newcommand{\ih}{\textit{i.h.}}

\newcommand*\coloneq{
 \mathrel{%
  \rlap{\raisebox{0.3ex}{$\m@th\cdot$}}%
        \raisebox{-0.3ex}{$\m@th\cdot$}%
 {=}}}
\newcommand*\coloneqq{
 \mathrel{%
  \rlap{\raisebox{0.3ex}{$\m@th\cdot$}}%
        \raisebox{-0.3ex}{$\m@th\cdot$}%
  \rlap{\raisebox{0.3ex}{$\m@th\cdot$}}%
        \raisebox{-0.3ex}{$\m@th\cdot$}%
 {=}}}

\newcommand\vc[1]{\vcenter{\hbox{$#1$}}}

\newcommand\smallbin[1]{\mathchoice
      {\mathbin{\raise.2ex \hbox{$\scriptstyle      #1$}}}%
      {\mathbin{\raise.2ex \hbox{$\scriptstyle      #1$}}}%
      {\mathbin{\raise.12ex\hbox{$\scriptscriptstyle#1$}}}%
      {\mathbin{           \hbox{$\scriptscriptstyle#1$}}}}%

\newcommand\smallsquare{\mathchoice{{\scriptstyle\square}}{{\scriptstyle\square}}{{\scriptscriptstyle\square}}{{\scriptscriptstyle\square}}}

\newcommand\Con{\wedge}
\newcommand\Imp{\rightarrow}

\newcommand\con{\kern1pt{\smallbin\Con}\kern1pt}
\newcommand\imp{\kern1pt{\smallbin\Imp}}

\newcommand\PEL{\Lambda_{\textsf{PE}}}

\newcommand\fv[1]{\mathsf{fv}(\trm{#1})}
\newcommand\fl[1]{\mathsf{fl}(\trm{#1})}

\colorlet{mgray}{black!40}
\colorlet{lgray}{black!25}
\colorlet{llgray}{black!15}

\colorlet{dblue}{blue!80!black}
\colorlet{dred}{red!80!black}

\colorlet{typecolor}{dblue}
\colorlet{termcolor}{dred}

\newcommand\typecolor{\color{typecolor}}
\newcommand\termcolor{\color{termcolor}}

\newcommand\black{\color{black}}

\newcommand\type[1]{{\let\type@sup@color\termcolor\typecolor\typ{#1}}}
\newcommand\typ[1]{%
  \let\type@loop=\type@next%
  \type@loop#1,%
}
\newcommand\type@next[1]{%
  \ifx#1,\let\type@loop\type@end\else%
  \ifx#1_\let\type@loop\type@sub\else%
  \ifx#1^\let\type@loop\type@sup\else%
  \ifx#1*\con\else%
  \ifx#1-\kern1pt{\imp}\else%
  #1%
  \fi\fi\fi\fi\fi%
  \type@loop%
}
\newcommand\type@sup@color{}
\newcommand\type@sub[1]{_{#1}\let\type@loop\type@next\type@loop}
\newcommand\type@sup[1]{^{{\type@sup@color #1}}\let\type@loop\type@next\type@loop}
\newcommand\type@vec[1]{\vec{\kern.5pt#1\kern.5pt}\let\type@loop\type@next\type@loop}
\newcommand\type@end{\let\type@sup@color\relax}

\newcommand\x{\lambda x}

\newcommand\+[1][{}]{\kern1pt{\smallbin\oplus}_{#1}\kern1pt}

\newcommand\lab{\bullet}

\newcommand\ttrm[1]{\smash{\trm{#1}}}

\newcommand\term[1]{{\let\term@typecolor\typecolor\termcolor\trm{#1}}}
\newcommand\trm[1]{%
  \vphantom(%
  \let\term@loop=\term@next%
  \term@loop#1,%
}
\newcommand\term@next[1]{%
  \ifx#1,\let\term@loop\term@end\else%
  \ifx#1:\black\colon\term@typecolor\let\term@loop\term@type\else%
  \ifx#1_\let\term@loop\term@sub\else%
  \ifx#1^\let\term@loop\term@sup\else%
  \ifx#1!\let\term@loop\term@box\else%
  \ifx#1+\let\term@loop\term@prob\else%
  \ifx#1*^\lab\else%
  \ifx#1<\lfloor\else%
  \ifx#1>\rfloor\else%
  \ifx#1..\,\else%
  \ifx#1=\kern1pt{\smallbin=}\kern1pt\else
  #1%
  \fi\fi\fi\fi\fi\fi\fi\fi\fi\fi\fi
  \term@loop%
}
\newcommand\term@typecolor{}
\newcommand\term@end{\let\term@typecolor\relax}
\newcommand\term@sub[1]{_{#1}\let\term@loop\term@next\term@loop}
\newcommand\term@sup[1]{^{#1}\let\term@loop\term@next\term@loop}
\newcommand\term@vec[1]{\vec{\kern.5pt#1\kern.5pt}\let\term@loop\term@next\term@loop}
\newcommand\term@prob[1]{\kern1pt\raisebox{-.5pt}{$\overset{\raisebox{-1pt}{$\scriptstyle#1$}}{{\smallbin\oplus}}$}\kern1pt\let\term@loop\term@next\term@loop}
\newcommand\term@type{\let\type@loop=\type@next\type@loop}
\newcommand\term@box[1]{\probox{#1}\let\term@loop\term@next\term@loop}

\newcommand\probox[1]{\begin{tikzpicture}[baseline=0]\node[anchor=base](a){$\scriptstyle #1\vphantom)$};\draw[line width=.6pt] (-5pt,-2.5pt) rectangle (5pt,7.5pt);\end{tikzpicture}}

\newcommand{\labjudg}[2]{#1\vdash_{L} #2}

\newcommand\rw[1][{}]{\stackrel{#1}\rightsquigarrow}

\newcommand\perm{\mathsf p}

\newcommand\rstrut{\rule{0pt}{13pt}}

\newcommand\SN{\textsf{SN}}

\newcommand\proj[3]{\pi^{#1}_{#2}(\trm{#3})}

\newcommand\cbn{\mathsf{cbn}}
\newcommand\cbv{\mathsf{cbv}}

\newcommand\uncbv[1]{\llbracket#1\rrbracket_\cbv}
\newcommand\unopen[1]{\llbracket#1\rrbracket_{\mathsf{open}}}
\newcommand\labclose[2]{\lfloor\labjudg{#1}{#2}\rfloor}

\tikzstyle{implied}=[dashed]
\tikzstyle{rwhead}=[>/.tip={Triangle[open,length=2.5pt,width=4.5pt]},|/.tip={Rectangle[length=.5pt,width=4.5pt]}]
\tikzstyle{rwblack}=[>/.tip={Triangle[length=2.5pt,width=4.5pt]},|/.tip={Rectangle[length=.5pt,width=4.5pt]}]
\tikzstyle{rw} =[line width=.5pt,rwhead,->]
\tikzstyle{rwb}=[line width=.5pt,rwblack,->]
\tikzstyle{rwbs}=[line width=.5pt,rwblack,->.>]
\tikzstyle{rws}=[line width=.5pt,rwhead,->.>]
\tikzstyle{rwn}=[line width=.5pt,rwhead,->.>|]
\tikzstyle{rwbn}=[line width=.5pt,rwblack,->.>|]
\tikzstyle{rwp}=[line width=.5pt,rwhead,->,double]
\tikzstyle{rwx}=[line width=.5pt,rwhead,->.>,double]

\renewcommand\rw{\mathrel{\tikz\draw[rw](0,0)--(10pt,0pt);}}
\newcommand\rwb{\mathrel{\tikz\draw[rwb](0,0)--(10pt,0pt);}}
\newcommand\rwbs{\mathrel{\tikz\draw[rwbs](0,0)--(10pt,0pt);}}
\newcommand\rws{\mathrel{\tikz\draw[rws](0,0)--(10pt,0pt);}}
\newcommand\rwn{\mathrel{\tikz\draw[rwn](0,0)--(10pt,0pt);}}
\newcommand\rwbn{\mathrel{\tikz\draw[rwbn](0,0)--(10pt,0pt);}}
\newcommand\rwp{\mathrel{\tikz\draw[rwp](0,0)--(10pt,0pt);}}
\newcommand\rwx{\mathrel{\tikz\draw[rwx](0,0)--(10pt,0pt);}}

\newcommand\rwpleft{\mathrel{\tikz\draw[rwp](10pt,0pt)--(0,0);}}

\newcommand\rwsleft{\mathrel{\tikz\draw[rws](10pt,0pt)--(0,0);}}

\newcommand\idem{\ensuremath{\mathsf i}}
\newcommand\cancelL{\ensuremath{\mathsf c_1}}
\newcommand\cancelR{\ensuremath{\mathsf c_2}}
\newcommand\plusAbs{\ensuremath{{\smallbin\oplus}\lambda}}
\newcommand\plusArg{\ensuremath{{\smallbin\oplus}\mathsf a}}
\newcommand\plusFun{\ensuremath{{\smallbin\oplus}\mathsf f}}
\newcommand\plusBox{\ensuremath{{\smallbin{\oplus}\smallsquare}}}
\newcommand\plusL{\ensuremath{{\smallbin\oplus}{\smallbin\oplus}_1}}
\newcommand\plusR{\ensuremath{{\smallbin\oplus}{\smallbin\oplus}_2}}
\newcommand\plusX{\ensuremath{{\smallbin\oplus}{\star}}}
\newcommand\boxVoid{\ensuremath{\not{\kern-2pt\smallsquare}}}
\newcommand\boxAbs{\ensuremath{\smallsquare\lambda}}
\newcommand\boxFun{\ensuremath{\smallsquare\mathsf f}}

\makeatother

\title{Decomposing Probabilistic Lambda-calculi}

\author{
	 Ugo Dal Lago\inst1
\and Giulio Guerrieri\inst2
\and Willem Heijltjes\inst2
}

\authorrunning{U.\ Dal Lago et al.}

\institute{%
Dipartimento di Informatica - Scienza e Ingegneria\\ Universit\`a di Bologna, Bologna, Italy\\
\email{ugo.dallago@unibo.it}\\[10pt]
\and%
Department of Computer Science\\ University of Bath, Bath, UK\\
\email{\{w.b.heijltjes,g.guerrieri\}@bath.ac.uk}
}

\begin{document}

\maketitle

\begin{abstract}
A notion of probabilistic lambda-calculus usually comes with a prescribed reduction strategy, typically call-by-name or call-by-value, as the calculus is non-confluent and these strategies yield different results. This is a break with one of the main advantages of lambda-calculus: confluence, which means results are independent from the choice of strategy.
We present a probabilistic lambda-calculus where the probabilistic operator is decomposed into two syntactic constructs: a generator, which represents a probabilistic event; and a consumer, which acts on the term depending on a given event. The resulting calculus, the Probabilistic Event Lambda-Calculus, is confluent, and interprets the call-by-name and call-by-value strategies through different interpretations of the probabilistic operator into our generator and consumer constructs.
We present two notions of reduction, one via fine-grained local rewrite steps, and one by generation and consumption of probabilistic events. Simple types for the calculus are essentially standard, and they convey strong normalization. We demonstrate how we can encode call-by-name and call-by-value probabilistic evaluation.
\end{abstract}

\section{Introduction}

Probabilistic lambda-calculi \cite{SahebDjahromi78,Manber-Tompa-1982,JonesPlotkin89,deLiguoroPiperno95,JungTix98,DalLagoZorzi12,FaggianRonchi19} extend the lambda-calculus with a probabilistic choice operator $N\+[p]M$, which chooses $N$ with probability $p$ and $M$ with probability $1-p$ (throughout this paper, we let $p=0.5$ and will omit it). Duplication of $N\+M$, as is wont to happen in lambda-calculus, raises a fundamental question about its semantics: do the duplicate occurrences represent \emph{the same} probabilistic event, or \emph{different} ones with the same probability? For example, take the formula $\top\+\bot$ that represents a coin flip between boolean values \emph{true} $\top$ and \emph{false} $\bot$. If we duplicate this formula, do the copies represent two distinct coin flips with possibly distinct outcomes, or do these represent a single coin flip that determines the outcome for both copies? Put differently again, when we duplicate $\top\+\bot$, do we duplicate the \emph{event}, or only its \emph{outcome}?

In probabilistic lambda-calculus, these two interpretations are captured by the evaluation strategies of call-by-name ($\rw_\cbn$), which duplicates events, and call-by-value ($\rw_\cbv$), which evaluates any probabilistic choice before it is duplicated, and thus only duplicates outcomes. Consider the following example, where $=$ tests equality of boolean values.
\[
	\top \quad {}_\cbv\!\rwsleft \quad (\x.\,x = x)(\top\+\bot) \quad \rws_\cbn \quad \trm{\top\+\bot} 
\]
This situation is not ideal, for several, related reasons. First, it demonstrates how probabilistic lambda-calculus is non-confluent, negating one of the central properties of the lambda-calculus, and one of the main reasons why it is the prominent model of computation that it is. Second, a probabilistic lambda-calculus must derive its semantics from a prescribed reduction strategy, and its terms only have meaning in the context of that strategy. Third, combining different kinds of probabilities becomes highly involved~\cite{FaggianRonchi19}, as it would require specialized reduction strategies. These issues present themselves even in a more general setting, namely that of commutative (algebraic) effects, which in general do not commute with copying.

We address these issues by a decomposition of the probabilistic operator into a \emph{generator} $\ttrm{!a}$ and a \emph{choice} $\ttrm{+a}$, as follows.
\[
	\trm{N \+ M} \quad\stackrel\Delta=\quad \trm{!a. N +a M}
\]
Semantically, $\ttrm{!a}$ represents a probabilistic event, that generates a boolean value recorded as $a$. The choice $\ttrm{N+aM}$ is simply a conditional on $a$, choosing $N$ if $a$ is false and $M$ if $a$ is true. Syntactically, $a$ is a boolean variable with an occurrence in $\ttrm{+a}$, and $\ttrm{!a}$ acts as a quantifier, binding all occurrences in its scope. (To capture a non-equal chance, one would attach a probability $p$ to a generator, as $\ttrm{!a}{\kern1pt}_p$, though we will not do so in this paper.) 

The resulting \emph{probabilistic event lambda-calculus} $\PEL$, which we present in this paper, is confluent. Our decomposition allows us to separate duplicating an \emph{event}, represented by the generator $\ttrm{!a}$, from duplicating only its \emph{outcome} $a$, through having multiple choice operators $\trm{+a}$. In this way our calculus may interpret both original strategies, call-by-name and call-by-value, by different translations of standard probabilistic terms into $\PEL$. For the above example, we get the following translations and reductions.
\[
\begin{array}{l@{\quad}l@{\quad}l@{\quad}l@{\quad}l@{\quad}l}
	\cbn: & \trm{(\x.x = x)(!a.\top+a\bot)} & \rw_\beta & \trm{(!a.\top+a\bot)=(!b.\top+b\bot)} & \rws & \top\+\bot
\\[10pt]
	\cbv: & \trm{!a.(\x.x = x)(\top+a\bot)} & \rw_\beta & \trm{!a.(\top+a\bot)=(\top+a\bot)}    & \rws & \top
\end{array}
\]
In this paper, we introduce $\PEL$ and its reduction mechanisms (Sections~\ref{sec:PEL},~\ref{sec:p-reduction},~\ref{sec:projective-reduction}); we prove confluence (Section~\ref{sec:confluence}); we give a system of simple types and prove strong normalization for typed terms (Section~\ref{sec:SN}); and we demonstrate the translation to interpret call-by-value evaluation (Section~\ref{sec:cbv}).

\subsection{Related work}

Probabilistic $\lambda$-calculi are a topic of study since the pioneering work by Saheb-Djaromi~\cite{SahebDjahromi78}, the first to give the syntax and operational semantics of a $\lambda$-calculus with binary probabilistic choice. Giving well-behaved denotational models for probabilistic $\lambda$-calculi has proved to be challenging, as witnessed by the many contributions spanning the last thirty years: from Jones and Plotkin early study of the probabilistic powerdomain~\cite{JonesPlotkin89}, to Jung and Tix's remarkable (and mostly negative) observations~\cite{JungTix98}, to the very recent encouraging results by Goubault-Larrecq~\cite{GoubaultLarrecq19}. A particularly well-behaved model for probabilistic $\lambda$-calculus can be obtained by taking a probabilistic variation of Girard's coherent spaces~\cite{DanosEhrhard11}, this way getting full abstraction~\cite{EPT18}.

On the operational side, one could mention a study about the various ways the operational semantics of a calculus with binary probabilistic choice can be specified, namely by small-step or big-step semantics, or by inductively or coinductively defined sets of rules~\cite{DalLagoZorzi12}. Termination and complexity analysis of higher-order probabilistic programs seen as $\lambda$-terms have been studied by way of type systems in a series of recent results about size~\cite{DalLagoGrellois19}, intersection~\cite{BreuvartDalLago18}, and refinement type disciplines \cite{AvanziniDalLagoGhyselen19}. Contextual equivalence on probabilistic $\lambda$-calculi has been studied, and compared with equational theories induced by B\"ohm Trees~\cite{Leventis18}, applicative bisimilarity~\cite{DalLagoSangiorgiAlberti14}, or environmental bisimilarity~\cite{SangiorgiVignudelli16}.

In all the aforementioned works, probabilistic $\lambda$-calculi have been taken as implicitly endowed with either call-by-name or call-by-value strategies, for the reasons outlined above. There are only a few exceptions, namely some works on Geometry of Interaction~\cite{DLFVY17}, Probabilistic Coherent Spaces~\cite{EhrhardTasson19}, and Standardization~\cite{FaggianRonchi19}, which achieve, in different contexts, a certain degree of independence from the underlying strategy, thus accomodating both call-by-name and call-by-value evaluation. The way this is achieved, however, invariably relies on Linear Logic or related concepts. This is deeply different from what we do here.


Our permutative reduction implements the equational theory of \emph{(ordered) binary decision trees} via rewriting~\cite{Zantema-Pol-2001}. Probabilistic decision trees have been proposed with a primitive binary probabilistic operator~\cite{Manber-Tompa-1982}, but not a decomposition as we explore here.


\section{\texorpdfstring{The probabilistic event $\lambda$-calculus $\PEL$}{The probabilistic event lambda-claculus PEL}}
\label{sec:PEL}

\begin{defn}
The \emph{probabilistic event $\lambda$-calculus} ($\PEL$) is given by the following grammar, with from left to right: a \emph{variable}, an \emph{abstraction}, an \emph{application}, a \emph{(labelled) choice}, and a \emph{(probabilistic) generator}.
\[
	M,N \quad\coloneqq\quad x ~\mid~ \x.N ~\mid~ NM ~\mid~ \trm{N +a M} ~\mid~ \trm{!a.N}
\]
\end{defn}
In a term $\trm{\x.M}$ the abstraction $\x$ binds the free occurrences of the variable $x$ in its scope $\trm{M}$, and in $\ttrm{!a.N}$ the generator $\ttrm{!a}$ binds the \emph{label} $a$ in $\trm{M}$. The calculus features a decomposition of the usual probabilistic sum $\trm{\+}$, as follows.
\[
	N\+M \quad\stackrel\Delta=\quad \trm{!a. N +a M}
\]
The generator $\ttrm{!a}$ represents a probabilistic \emph{event}, whose outcome, a binary value $\{0,1\}$ represented by the label $a$, is used by the choice operator $\ttrm{+a}$. That is, $\ttrm{!a}$ flips a coin setting $a$ to $0$ or $1$, and depending on this $\ttrm{N+aM}$ reduces to $N$ respectively $M$. We will use the unlabelled choice $\+$ as the above abbreviation. This convention also gives the translation from a \emph{call-by-name} probabilistic lambda-calculus into $\PEL$ (we formalize the interpretation of a \emph{call-by-value} probabilistic calculus in Section~\ref{sec:cbv}).

\begin{figure}[!ht]
  \fbox{
\begin{minipage}{.97\textwidth}
\begin{align}
	(\x.N)M 				&\rw_\beta N[M/x]													\tag{$\beta$}
\\																								\notag
\\	\trm{N +a N}			&\rw_\perm N														\tag{\idem}
\\	\trm{(N +a M) +a P}		&\rw_\perm \trm{N +a P}					\rstrut						\tag{\cancelL}
\\	\trm{N +a (M +a P)}		&\rw_\perm \trm{N +a P}					\rstrut						\tag{\cancelR}
\\																								\notag
\\	\trm{\x.(N +a M)}		&\rw_\perm \trm{(\x.N) +a (\x.M)}									\tag{\plusAbs}
\\	\trm{(N +a M) P}		&\rw_\perm \trm{(NP) +a (MP)}			\rstrut						\tag{\plusFun}
\\	\trm{N (M +a P)}		&\rw_\perm \trm{(NM) +a (NP)}			\rstrut						\tag{\plusArg}
\\	\trm{(N +a M) +b P}		&\rw_\perm \trm{(N +b P) +a (M +b P)} 	\rstrut	&& (a\smallbin<b)	\tag{\plusL}
\\	\trm{N +b (M +a P)}		&\rw_\perm \trm{(N +b M) +a (N +b P)} 	\rstrut	&& (a\smallbin<b)	\tag{\plusR}
\\	\trm{!b.(N +a M)}		&\rw_\perm \trm{(!b.N) +a (!b.M)}		\rstrut	&& (a\neq b)		\tag{\plusBox}
\\																								\notag
\\	\trm{!a.N}				&\rw_\perm N 									&& (a\notin N)		\tag{\boxVoid}
\\	\trm{\x.!a.N} 			&\rw_\perm \trm{!a.\x. N}				\rstrut						\tag{\boxAbs}
\\	\trm{(!a.N)M}			&\rw_\perm \trm{!a.(NM)}				\rstrut						\tag{\boxFun}
\end{align}
\end{minipage}}
\caption{Reduction rules}
\label{fig:reduction rules}
\end{figure}

\subsection{Reduction}

Reduction in $\PEL$ will consist of standard $\beta$-reduction $\rw_\beta$ plus an evaluation mechanism for generators and choice operators, which implements probabilistic choice. We will present two such mechanisms: \emph{projective} reduction~$\rw_\pi$ and \emph{permutative} reduction~$\rw_\perm$. While projective reduction implements the given intuition for the generator and choice operator, we relegate it to Section~\ref{sec:projective-reduction} and make permutative reduction our main evaluation mechanism, for the reason that it is more fine-grained, and thus more general. 

Permutative reduction is based on the idea that any operator distributes over the labelled choice operator (see the reduction steps in Figure~\ref{fig:reduction rules}), even other choice operators, as below.
\[
	\trm{(N +a M) +b P}	~\sim~ \trm{(N +b P) +a (M +b P)}
\]
To orient this as a rewrite rule, we need to give priority to one label over another. Fortunately, the relative position of the associated generators $\ttrm{!a}$ and $\ttrm{!b}$ provides just that. Then to define $\rw_\perm$, we will want every choice to belong to some generator, and make the order of generators explicit.

\begin{defn}
	The set $\fl{N}$ of \emph{free labels} of a term $\trm{N}$ is defined inductively by:
	\begin{align*}
		\fl{x} &= \emptyset & \fl{MN} &= \fl{M} \cup \fl{N}  & \fl{\x.M} &= \fl{M} \\
		\fl{!a.M} &= \fl{M} \smallsetminus \{a\} & \fl{M +a N} &= \fl{M}\cup \fl{N} \cup \{a\}
	\end{align*}
	A term $\trm{M}$ is \emph{label-closed} if $\fl{M} = \emptyset$.
\end{defn}

From now on, we consider only label-closed terms (and we implicitly assume this, unless otherwise stated).
All terms are identified up to renaming of their bound variables and labels.
Given some terms $\trm{M}$ and $\trm{N}$ and a variable $\trm{x}$, $\trm{M[N/x]}$ stands for the capture-avoiding (for both variables and labels) substitution of $\trm{N}$ for the free occurrences of $\trm{x}$ in $\trm{M}$.
We talk of a \emph{representative} $\trm{M}$ of a term when $\trm{M}$ is not considered up to such a renaming.
A representative $\trm{M}$ of a term is \emph{well-labeled} if
for every occurrence of $\probox a$ in $\trm{M}$ there is no $\probox a$ occurring in its scope.

\begin{defn}[Order for labels]
\label{def:orderlabels}
	Let $\trm{M}$ be a well-labeled representative of a term.
	We define an \emph{order} $<_{{\trm{M}}}$ for the labels occurring in $\trm{M}$ as follows: $a <_{\trm{M}} b$ if and only if $\probox b$ occurs in the scope of $\probox a$.
\end{defn}

\noindent
For a well-labeled and label-closed representative $\trm{M}$, $<_{\trm{M}}$ is a finite tree order.

\begin{defn}
\emph{Reduction} $\rw \,=\, \rw_\beta \cup \rw_\perm$ in $\PEL$ consists of \emph{$\beta$-reduction}~$\rw_\beta$ and \emph{permutative} or \emph{$\perm$-reduction}~$\rw_\perm$, both given in Figure~\ref{fig:reduction rules}. We write $\rws$ for the reflexive-transitive closure of reduction and $\rwn$ for reduction to normal form, and similarly for $\rw_\beta$ and $\rw_\perm$.
\end{defn}

\noindent
The introduction briefly sketches two example reductions; a third, complete reduction is given in Figure~\ref{fig:example reduction}.
The crucial feature of $\perm$-reduction is that a choice $\ttrm{+a}$ \emph{does} permute out of the argument position of an application, but a generator $\ttrm{!a}$ does \emph{not}, as below. Since the argument of a redex may be duplicated, this is how we characterize the difference between the \emph{outcome} of a probabilistic event, whose duplicates may be identified, and the event itself, whose duplicates may yield different outcomes.
\[
	\trm{N\,(M +a P)}~\rw_\perm~ \trm{(NM) +a (NP)} 
	\qquad
	\qquad
	\trm{N\,(!a.M)}~\not\rw_\perm~\trm{!a.N\,M}
\]
By inspection of the rewrite rules in Figure~\ref{fig:reduction rules}, we can then characterize the normal forms of $\rw_\perm$ and $\rw$ as follows.

\begin{prop}[Syntactic characterization of normal forms]
The normal forms $P_0$ of $\rw_\perm$, respectively $N_0$ of $\rw$, are given by the following grammars.

\[
	\begin{array}{ccc@{~}c@{~}c@{}l}
		P_0 &\coloneqq& P_1 &\mid& \trm{P_0 \+ P_0}
	\\	P_1	&\coloneqq& x   &\mid& \x.P_1 			 &~\mid~ P_1\,P_0
	\end{array}
	\qquad\qquad
	\begin{array}{ccc@{~}c@{~}c}
		N_0 &\coloneqq& N_1 &\mid& \trm{N_0 \+ N_0}
	\\	N_1	&\coloneqq& N_2 &\mid& \x.N_1
	\\	N_2 &\coloneqq& x	&\mid& N_2\,N_0
	\end{array}
\]
\end{prop}

\begin{figure}[!t]
\newcommand\fit[1]{\makebox[36pt][c]{$#1$}}%
\begin{align*}
	\trm{!a.(\x.x = x)(\top+a\bot)} 
			& \fit{\rw_\perm}	\trm{!a.(\x.x=x)\top~+a~(\x.x=x)\bot}	\tag\plusArg
\\		& \fit{\rws_\beta}	\trm{!a.(\top=\top)\,+a\,(\bot=\bot)}
\\		& \fit{=}			\trm{!a. \top +a \top}
			  \fit{\rw_\perm} 	\trm{!a. \top} 
			  \fit{\rw_\perm}	\top           						 	\tag{\idem,\boxVoid}
\end{align*}
\caption{Example reduction of the $\cbv$-translation of the term in the introduction.}
\label{fig:example reduction}
\end{figure}


\section{Properties of permutative reduction}
\label{sec:p-reduction}

We will prove strong normalization (\SN) and confluence of $\rw_\perm$. For strong normalization, the obstacle is the interaction between different choice operators, which may duplicate each other, and even themselves, creating super-exponential growth.\footnote{This was inferred from only a simple simulation; we would be interested to know a rigorous complexity result.} Fortunately, Dershowitz's \emph{recursive path orders}~\cite{Dershowitz82} seem tailor-made for our situation.


We observe that the set $\PEL$ endowed with $\rw_\perm$ is a (first-order) term rewriting system over a countably infinite set of variables (denoted by $x, y, z, \dots$) and the signature $\Sigma$ given by:
\begin{itemize}
	\item the binary function symbol $\trm{+a}$, for any label $a$;
	\item the unary function symbol $\trm{!a}$, for any label $a$;
	\item the unary function symbol $\trm{\x}$, for any variable $x$;
	\item the binary function symbol $\trm{@}$, letting $@(\trm{M},\trm{N})$ stand for $MN$.
\end{itemize}

\begin{defn}
	Let $\trm{M}$ be a well-labeled representative of a label-closed term, and let $\Sigma_M$ be the set of signature symbols occurring in $\trm{M}$.
	We define $\prec_M$ as the (strict) partial order on $\Sigma_M$ generated by the following rules.
\[
\begin{array}{rcl@{\qquad\quad}l}
		\trm{+a} &\prec_M& \trm{+b} & \text{ if } a <_M b
\\	\trm{+a} &\prec_M& \trm{!b} & \text{ for any labels } a,b
\\	\trm{!b} &\prec_M& @,\x		& \text{ for any label } b
\end{array}
\]
\end{defn}

\begin{lem}
\label{lemma:strong-normalization}
	The reduction $\rw_\perm$ is strongly normalizing.
\end{lem}

\begin{proof}
For the first-order term rewriting system $(\PEL, \rw_\perm)$ we derive a well-founded recursive path ordering $<$ from $\prec_M$ following \cite[p. 289]{Dershowitz82}. Let $f$ and $g$ range over function symbols, let $[N_1,\dots,N_n]$ denote a multiset and extend $<$ to multisets by the standard multiset ordering, and let $N = f(N_1,\dots,N_n)$ and $M = g(M_1,\dots,M_m)$; then
\[
N < M \iff
\left\{
\begin{array}{ll}
	[N_1,\dots,N_n] < [M_1,\dots,M_m] & \text{ if } f = g
\\[5pt]
	[N_1,\dots,N_n] < [M]			  & \text{ if } f \prec_M g 
\\[5pt]
	[N] \leq [M_1,\dots,M_m]		  & \text{ if } f \npreceq_M g~.
\end{array}
\right.
\]
While $\prec_M$ is defined only relative to $\Sigma_M$, reduction may only reduce the signature. Inspection of Figure~\ref{fig:reduction rules} then shows that $M \rw_\perm N$ implies $N<M$.
\end{proof}
Of course, there is not to prove $\rw$ being strongly normalizing,
due to the presence of $\beta$-reduction rules and the absence of types.

\subsection{Confluence of permutative reduction}

With strong normalization, confluence of $\rw_\perm$ requires only local confluence. We begin by reducing the number of cases to consider, by casting the permutations of $\ttrm{+a}$ as instances of a common shape.

\begin{defn}
We define a \emph{context} $C[\,]$ as follows.
\[
\begin{array}{lll@{~}l@{~}l@{~}l@{~}l}
	C[\,] &\coloneqq& [\,] &\mid& \lambda x.C[\,] &\mid& C[\,]M ~\mid~ NC[\,] ~\mid~ \trm{C[\,]+aM} ~\mid~ \trm{N+aC[\,]} ~\mid~ \trm{!a.C[\,]}
\end{array}
\]
The term $C[N]$ represents $C[\,]$ with the hole $[\,]$ replaced by $N$.
\end{defn}

Observe that the six reduction rules $\plusAbs$ through $\plusBox$ in Figure~\ref{fig:reduction rules} are all of the following form. We refer to these collectively as $\plusX$.
\begin{align}
	\trm{C[N+aM]} \rw_\perm \trm{C[N]+aC[M]}
	\tag\plusX
\end{align}

\newcounter{lem:confluence-perm}
\addtocounter{lem:confluence-perm}{\value{defn}}
\begin{lem}[Confluence of $\rw_\perm$]
	\label{lem:confluence-perm}
Reduction $\rw_\perm$ is confluent.
\end{lem}


\begin{proof}
We prove local confluence; by Newman's lemma and strong normalization of $\rw_\perm$ (\Cref{lemma:strong-normalization}), confluence follows. We consider each reduction rule against those lower down in Figure~\ref{fig:reduction rules}. For the symmetric rule pairs $\cancelL$/$\cancelR$, $\plusL$/$\plusR$, and $\plusFun/\plusArg$ we present only the first case. Unless otherwise specified, we let $a<b<c$.

\newcommand\itm[2]{\medskip\noindent(#1)}

\itm\idem{\trm{N +a N}\rw N}
{\small
\[
\vc{\begin{tikzpicture}[x=20pt,y=1.5ex]
	\node[anchor=base east] (a) at (0,0) {$\trm{(N+aN)+aM}$};
	\node[anchor=base west] (b) at (1,0) {$\trm{N+aM}$};
	\draw[rw] (0,1) --node[above]{$\idem$}    (1,1);
	\draw[rw] (0,0) --node[below]{$\cancelL$} (1,0);
\end{tikzpicture}}
\]
\[
\vcenter{\hbox{\begin{tikzpicture}
	\matrix [matrix of math nodes] (m) {
	  		\trm{(N+aM)+a(N+aM)} &[20pt] \trm{N+aM}
	\\[20pt]\trm{N+a(N+aM)}
	\\ };
	\draw[rw] (m-1-1) --node[above]{$\idem$}    (m-1-2);
	\draw[rw] (m-1-1) --node[left] {$\cancelL$} (m-2-1);
	\draw[rw,implied] (m-2-1) --node[below right=-2pt] {$\cancelR$} (m-1-2);
\end{tikzpicture}}}
\qquad
\vcenter{\hbox{\begin{tikzpicture}
	\matrix [matrix of math nodes] (m) {
	  		\trm{C[N+aN]}    &[20pt] C[N]
	\\[20pt]\trm{C[N]+aC[N]}
	\\ };
	\draw[rw] (m-1-1) --node[above]{$\idem$}  (m-1-2);
	\draw[rw] (m-1-1) --node[left] {$\plusX$} (m-2-1);
	\draw[rw,implied] (m-2-1) --node[below right=-2pt] {$\idem$} (m-1-2);
\end{tikzpicture}}}
\]
\[
\vcenter{\hbox{\begin{tikzpicture}
	\matrix [matrix of math nodes] (m) {
	  		\trm{(N+aM)+b(N+aM)} &[30pt] \trm{N+aM}
	\\[20pt]\trm{(N+b(N+aM))+a(M+b(N+aM))} 
	\\[20pt]\trm{((N+bN)+a(N+bM))+a((M+bN)+a(M+bM))} & \trm{(N+bN)+a(M+bM)}
	\\ };
	\draw[rw] (m-1-1) --node[above]{$\idem$}    (m-1-2);
	\draw[rw] (m-1-1) --node[left] {$\plusL$} (m-2-1);
	\draw[rws,implied] (m-2-1) --node[left] {$\plusR$}   (m-3-1);
	\draw[rws,implied] (m-3-2) --node[right]{$\idem$}    (m-1-2);
	\draw[rws,implied] (m-3-1) --node[below]{$\cancelL,\cancelR$} (m-3-2);
\end{tikzpicture}}}
\]
}

\itm\cancelL{\trm{(N +a M) +a P}\rw\trm{N +a P}}
{\small
\[
\vcenter{\hbox{\begin{tikzpicture}
	\matrix [matrix of math nodes] (m) {
	  		\trm{(N+aM)+a(P+aQ)} &[15pt] \trm{N+a(P+aQ)}
	\\[20pt]\trm{(N+aM)+a Q} & \trm{N+aQ}
	\\ };
	\draw[rw] (m-1-1) --node[above]{$\cancelL$} (m-1-2);
	\draw[rw] (m-1-1) --node[left] {$\cancelR$} (m-2-1);
	\draw[rw,implied] (m-1-2) --node[right]{$\cancelR$} (m-2-2);
	\draw[rw,implied] (m-2-1) --node[below]{$\cancelL$} (m-2-2);
\end{tikzpicture}}}
\vcenter{\hbox{\begin{tikzpicture}
	\matrix [matrix of math nodes] (m) {
	  		\trm{C[(N+aM)+aP]} &[15pt] \trm{C[N+aM]}
	\\[20pt]\trm{C[N+aM]+aC[P]}
	\\[20pt]\trm{(C[N]+aC[M])+aC[P]} & \trm{C[N]+aC[M]}
	\\ };
	\draw[rw] (m-1-1) --node[above]{$\cancelL$} (m-1-2);
	\draw[rw] (m-1-1) --node[left] {$\plusX$}   (m-2-1);
	\draw[rw,implied] (m-2-1) --node[left] {$\plusX$}   (m-3-1);
	\draw[rw,implied] (m-1-2) --node[right]{$\plusX$}   (m-3-2);
	\draw[rw,implied] (m-3-1) --node[below]{$\cancelL$} (m-3-2);
\end{tikzpicture}}}
\]
\[
\vcenter{\hbox{\begin{tikzpicture}
	\matrix [matrix of math nodes] (m) {
			\trm{(N+bM)+b(P+aQ)} &[20pt]  \trm{N+b(P+aQ)}
	\\[20pt]\trm{((N+bM)+bP)+a((N+bM)+bQ)} & \trm{(N+bP)+a(N+bQ)}
	\\ };
	\draw[rw] (m-1-1) --node[above]{$\cancelL$} (m-1-2);
	\draw[rw] (m-1-1) --node[left] {$\plusR$}   (m-2-1);
	\draw[rw, implied] (m-1-2) --node[right]{$\plusR$} (m-2-2);
	\draw[rws,implied] (m-2-1) --node[below]{$\cancelL$} (m-2-2);
\end{tikzpicture}}}
\]
}

\itm\plusX{\trm{C[N+aM]}\rw\trm{C[N]+a C[M]}}
{\small
\[
\vcenter{\hbox{\begin{tikzpicture}
	\matrix [matrix of math nodes] (m) {
	  		\trm{C[(N+aM)+bP]}      &[20pt] \trm{C[N+aM]+bC[P]}
	\\[20pt]\trm{C[(N+bP)+a(M+bP)]} &       \trm{(C[N]+aC[M])+bC[P]}		 
	\\[20pt]\trm{C[N+bP]+aC[M+bP]}  &       \trm{(C[N]+bC[P])+a(C[M]+bC[P])}
	\\ };
	\draw[rw] (m-1-1) --node[above]{$\plusX$} (m-1-2);
	\draw[rw] (m-1-1) --node[left] {$\plusL$} (m-2-1);
	\draw[rw, implied] (m-1-2) --node[right]{$\plusX$} (m-2-2);
	\draw[rw, implied] (m-2-1) --node[left] {$\plusX$} (m-3-1);
	\draw[rw, implied] (m-2-2) --node[right]{$\plusL$} (m-3-2);
	\draw[rws,implied] (m-3-1) --node[below]{$\plusX$} (m-3-2);
\end{tikzpicture}}}
\]

\itm\plusL{\trm{(N+aM)+bP}\rw\trm{(N+bP)+a(M+bP)}}
\[
\vcenter{\hbox{\begin{tikzpicture}[x=340pt,y=40pt]
	\node[anchor=west] (a) at (0,2) {$\trm{(N+aM)+b(P+aQ)}$};
	\node[anchor=west] (b) at (0,1) {$\trm{((N+aM)+bP)+a((N+aM)+bQ)}$};
	\node[anchor=west] (c) at (0,0) {$\trm{((N+bP)+a(M+bP))+a((N+bQ)+a(M+bQ))}$};
	\node[anchor=east] (d) at (1,2) {$\trm{(N+b(P+aQ))+a(M+b(P+aQ))}$};
	\node[anchor=east] (e) at (1,1) {$\trm{((N+bP)+a(N+bQ))+a((M+bP)+a(M+bQ))}$};
	\node[anchor=east] (f) at (1,0) {$\trm{(N+bP)+a(M+bQ)}$};
	\draw[rw] (a) --node[above]{$\plusL$} (d);
	\draw[rw] 			($(a.south west)+( 43.5pt,0pt)$) --node[left] {$\plusR$} ($(b.north west)+( 43.5pt,0pt)$);
	\draw[rws,implied]	($(d.south east)+(-43.5pt,0pt)$) --node[right]{$\plusR$} ($(e.north east)+(-43.5pt,0pt)$);
	\draw[rws,implied]	($(b.south west)+( 43.5pt,0pt)$) --node[left] {$\plusL$} ($(c.north west)+( 43.5pt,0pt)$);
	\draw[rws,implied]	($(e.south east)+(-43.5pt,0pt)$) --node[right]{$\cancelL,\cancelR$} ($(f.north east)+(-43.5pt,0pt)$);
	\draw[rws,implied] (c) --node[below]{$\cancelL,\cancelR$} (f);
\end{tikzpicture}}}
\]
\[
\vcenter{\hbox{\begin{tikzpicture}
	\matrix [matrix of math nodes] (m) {
	  		\trm{(N+bM)+c(P+aQ)}            &[20pt] \trm{(N+c(P+aQ))+b(M+c(P+aQ))}
	\\[20pt]                                &       \trm{((N+cP)+a(N+cQ))+b((M+cP)+a(M+cQ))} 
	\\[20pt]\trm{((N+bM)+cP)+a((N+bM)+cQ)}  &       \trm{((N+cP)+b(M+cP))+a((N+cQ)+b(M+cQ))}
	\\ };
	\draw[rw] (m-1-1) --node[above]{$\plusL$} (m-1-2);
	\draw[rw] (m-1-1) --node[left] {$\plusR$} (m-3-1);
	\draw[rws,implied] (m-1-2) --node[right]{$\plusR$} (m-2-2);
	\draw[rws,implied] (m-2-2) --node[right]{$\plusL,\plusR,\cancelL,\cancelR$} (m-3-2);
	\draw[rws,implied] (m-3-1) --node[below]{$\plusL$} (m-3-2);
\end{tikzpicture}}}
\]
}

The remaining cases are considered in the proof in the Appendix (p.~\pageref{lemmaAppendix:confluence-perm}).
\qedhere

\end{proof}

\begin{defn}
We denote the unique $\perm$-normal form of a term $N$ by $N^\perm$.
\end{defn}


\section{Confluence}
\label{sec:confluence}

We aim to prove that $\rw \,=\, \rw_\beta \cup \rw_\perm$ is confluent. We will use the standard technique of \emph{parallel} $\beta$-reduction, a simultaneous reduction step on an arbitrary number of $\beta$-redexes in the source term, which we define via a labeling of the redexes to be reduced. The central point is to find a notion of reduction that is \emph{diamond}, \ie\ every critical pair can be closed in one (or zero) steps. This will be our \emph{complete} reduction, which consists of parallel $\beta$-reduction followed by $\perm$-reduction to normal form.

\begin{defn}
A \emph{labeled} term $\trm{N*}$ is a term $N$ with chosen $\beta$-redexes annotated as $\trm{(\x.M)*P}$. The \emph{labeled reduct} $\trm{<N*>}$ of a labeled term is defined by induction on $N$ as follows:
\begin{align*}
	\trm{<(\x.N*)*M*>} &= \trm{<N*>[<M*>/x]}	&	\trm{<N*M*>} 	&= \trm{<N*><M*>}
\\	\trm{<x>}		&= x						&	\trm{<N*+aM*>}	&= \trm{<N*>+a<M*>}
\\	\trm{<\x.N*>}	&= \trm{\x.<N*>}			&	\trm{<\,!a.N*>}	&= \trm{!a.<N*>}
\end{align*}
A \emph{parallel $\beta$-step} is a reduction $N\rwp_\beta\trm{<N*>}$ for some labelling $\trm{N*}$ of $N$.
\end{defn}

\noindent
We write $\trm{N*}\rwp_\beta\trm{<N*>}$ for the specific parallel step indicated by the labeling $\trm{N*}$. Observe that $\trm{<N*>}$ is a regular unlabeled term, since all labels are removed in the reduction. For the empty labelling, $\trm{N} = \trm{<N*>}$, so that parallel reduction is reflexive: $\trm{N} \rwp_\beta \trm{N}$.

\begin{lem}
A parallel $\beta$-step $N\rwp_\beta M$ is a $\beta$-reduction $N\rws_\beta M$.
\end{lem}

\begin{proof}
By induction on the labelled term $\trm{N*}$ generating $N\rwp_\beta\trm{<N*>}=M$.
\end{proof}

For the commutation of (parallel) $\beta$-reduction with $\perm$-reduction, we run into the minor issue that a permuting generator or choice operator may block a redex: in both cases below, on the left the term has a redex, but on the right it is blocked. We address by an adaptation of $\perm$-reduction $\rwb_\perm$ on labelled terms, which is a strategy in $\rws_\perm$ that permutes past a labelled redex in one step.
\[
\begin{array}{rcl}
	\trm{(\x.N+aM)\,P} & \rw_\perm & \trm{((\x.N)+a(\x.M))\,P}
\\[5pt]
	\trm{(\x.!a.N)\,M} & \rw_\perm & \trm{(!a.\x.N)\,M}
\end{array}
\]

\begin{defn}
A \emph{labelled} $\perm$-reduction $\trm{N*}\rwb_\perm\trm{M*}$ is a $\perm$-reduction of one of the forms
\[
\begin{array}{rcl}
	\trm{(\x.N*+aM*)*P*} &\rws_\perm& \trm{(\x.N*)*P*+a(\x.M*)*P*}
\\[5pt]
	\trm{(\x.!a.N*)*M*} &\rws_\perm& \trm{!a.(\x.N*)*M*}
\end{array}
\]
or a single $\perm$-step $\rw_\perm$ on unlabeled constructors in $\trm{N*}$.
\end{defn}

\begin{lem}
\label{lem:parallel p-reduction}
Reduction to normal form in $\rwb_\perm$ is equal to $\rwn_\perm$.
\end{lem}

\begin{proof}
Observe that $\rw_\perm$ and $\rwb_\perm$ have the same normal forms. Then since $\rwb_\perm\,\subseteq\,\rws_\perm$ we have $\rwbn_\perm\,\subseteq\,\rwn_\perm$. Conversely, let $N\rwn_\perm M$. On this reduction, let $P\rw_\perm Q$ be the first step such that $P\not\rwb_\perm Q$. Then there is an $R$ such that $P\rwb_\perm R$ and $Q\rw_\perm R$. Note that we have $N\rwbs_\perm R$. By confluence, $R\rwn_\perm M$, and by induction on the sum length of paths in $\rw_\perm$ from $R$ (smaller than from $N$) we have $R\rwbn_\perm M$, and hence $N\rwbn_\perm M$.
\end{proof}

The following lemmata then give the required commutation properties of the relations $\rwb_\perm$, $\rwn_\perm$, and $\rwp_\beta$. Figure~\ref{fig:confluence diagrams} illustrates these by commuting diagrams.

\begin{lem}
\label{lem:parallel p - parallel beta}
If $\trm{N*}\rwb_\perm\trm{M*}$ then $\trm{<N*>}=_\perm\trm{<M*>}$.
\end{lem}

\begin{proof}
By induction on the rewrite step $\rwb_\perm$. The two interesting cases are:
\[
\begin{array}{cl}
\vcenter{\hbox{\begin{tikzpicture}
	\matrix[matrix of math nodes] (m) {
	  		\trm{(\x.M*)*(N* +a P*)} &[20pt] \trm{((\x.M*)*N*) +a ((\x. M*)*P*)}
	\\[20pt]\trm{<M*>[(<N*> +a <P*>)/x]} & \trm{<M*>[<N*>/x] +a <M*>[<P*>/x]}
	\\ };
	\draw[rwb] (m-1-1) --node[above]{$\perm$} (m-1-2);
	\draw[rwp] (m-1-1) --node[left] {$\beta$} (m-2-1);
	\draw[rwp,implied] (m-1-2) --node[right]{$\beta$} (m-2-2);
	\draw[rws,implied] (m-2-1) --node[below]{$\perm$} (m-2-2);
\end{tikzpicture}}}
&	(x\in\fv M)
\\ \\
\vcenter{\hbox{\begin{tikzpicture}
	\matrix[matrix of math nodes] (m) {
	  		\trm{(\x.M*)*(N* +a P*)} &[20pt] \trm{((\x.M*)*N*) +a ((\x. M*)*P*)}
	\\[20pt]\trm{<M*>}\vphantom{\trm{+a}} & \trm{<M*> +a <M*>}
	\\ };
	\draw[rwb] (m-1-1) --node[above]{$\perm$} (m-1-2);
	\draw[rwp] (m-1-1) --node[left] {$\beta$} (m-2-1);
	\draw[rwp,implied] (m-1-2) --node[right]{$\beta$} (m-2-2);
	\draw[rw, implied] (m-2-2) --node[below]{$\perm$} (m-2-1);
\end{tikzpicture}}}
&	(x\notin\fv M)
\end{array}
\qedhere
\]
\end{proof}

\begin{lem}
\label{lem:exhaustive p - parallel beta}
If $\trm{N*}\rwn_\perm\trm{M*}$ then $\trm{<N*>}=_\perm\trm{<M*>}$
\end{lem}

\begin{proof}
Using \Cref{lem:parallel p-reduction} we decompose $\trm{N*}\rwn_\perm\trm{M*}$ as
\[
	\trm{N*}=\trm{N_1*}\rwb_\perm\trm{N_2*}\rwb_\perm \dots \rwb_\perm \trm{N_n*}=\trm{M*}
\]
for which \Cref{lem:parallel p - parallel beta} gives $\trm{<N_i*>}=_\perm\trm{<N_{i+1}*>}$.
\end{proof}

\begin{defn}
A \emph{complete} reduction step $\rwp$ is a parallel $\beta$-step followed by $\perm$-reduction to normal form:
\[
	N\rwp M^\perm \quad\coloneq\quad N\rwp_\beta M\rwn_\perm M^\perm~.
\]
\end{defn}

\begin{lem}[Mapping on $\perm$-normal forms]
\label{lem:p to complete reduction}
Any reduction step $N\rw M$ maps onto a complete step $N^\perm\rwp M^\perm$ on $\perm$-normal forms.
\end{lem}

\begin{proof}
The case of a $\perm$-step $N\rw_\perm M$ is immediate, since $N^\perm=M^\perm$ and $\rwp_\beta$ is reflexive. 
For a $\beta$-step $N\rw_\beta M$ we label the reduced redex in $N$ to get $\trm{N*}\rwp_\beta\trm{<N*>}=M$. 
Let $\trm{P*}=(\trm{N*})^\perm$ be the (labeled) $\perm$-normal form of $\trm{N*}$. 
Then \Cref{lem:exhaustive p - parallel beta} gives $M=\trm{<N*>}=_\perm\trm{<P*>}$, illustrated below.
\[
\begin{tikzpicture}
	\matrix (m) [matrix of math nodes] {
	  \trm{N*} && \trm{P*} \\[20pt] \\ \trm{M=<N*>} & =_\perm & \trm{<P*>} \\
	};
	\draw[rwn] (m-1-1) --node[above]{$\scriptstyle\perm$} (m-1-3);
	\draw[rwp] (m-1-1) --node[left] {$\scriptstyle\beta$} (m-3-1);
	\draw[rwp] (m-1-3) --node[right]{$\scriptstyle\beta$} (m-3-3);
\end{tikzpicture}
\]
By confluence and strong normalization of $\rw_\perm$ then $\trm{<P*>}^\perm=M^\perm$. The complete reduction $N^\perm\rwp M^\perm$ is then given by $N^\perm\rwp_\beta\trm{<P*>}\rwn_\perm M^\perm$.
\end{proof}


Following \cite{Takahashi95}, the idea to prove that complete reduction $\rwp$ is diamond is: first, to prove that if $M \rwp_\beta N$ then $N \rwp_\beta M^*$, where $M^*$ is a term obtained from $M$ (and independent of $N$) by reducing all the $\beta$-redexes existing in $M$ simultaneously (\Cref{lemma:full-development}.\ref{lemma:full-development-beta} below); second, via \Cref{lem:p to complete reduction}, to lift this property to complete reduction $\rwp$ by computing the $\perm$-normal form of $M^*$ (\Cref{lemma:full-development}.\ref{lemma:full-development-complete} below).
To prove \Cref{lemma:full-development}.\ref{lemma:full-development-beta} we need the two following technical lemmata.

\newcounter{lemma:application-parallel-beta}
\addtocounter{lemma:application-parallel-beta}{\value{defn}}
\begin{lem}
\label{lemma:application-parallel-beta}
	If $\trm{M} \rwp_\beta \trm{M'}$ and $\trm{N} \rwp_\beta \trm{N'}$, 
\marginpar{\footnotesize Proof in the Appendix}
	then $\trm{MN} \rwp_\beta \trm{M'N'}$.
	If moreover $M = \lambda x.R$ and $M'^ = \lambda x R'$ with $R \rwp_\beta R'$, then  $\trm{MN} \rwp_\beta \trm{R'[N'/x]}$.
\end{lem}

%

\begin{lem}[Substitution]
\label{lemma:substitution}
	If $\trm{M} \rwp_\beta \trm{M'}$ and $\trm{N} \rwp_\beta \trm{N'}$, then $\trm{M<N/x>} \rwp_\beta \trm{M'<N'/x>}$.
\end{lem}

\begin{proof}
	By straightforward induction on $M$.
\end{proof}

\begin{lem}[Full labeling]
\label{lemma:full-development}
	Let $\trm{M}$ be a term, and $\trm{M*}$ be its \emph{full labeling}, which labels every $\beta$-redex in $\trm{M}$.
	\begin{enumerate}
		\item\label{lemma:full-development-beta} If $\trm{M} \rwp_\beta \trm{N}$ then $\trm{N} \rwp_\beta \trm{<M*>}$.
		\item\label{lemma:full-development-complete} If $\trm{M} \rwp \trm{N}$ then $\trm{N} \rwp \trm{<M*>}^\perm$.
	\end{enumerate}
\end{lem}
\begin{proof}\hfill
	\begin{enumerate}
		\item	By induction on $\trm{M}$.
		Cases:
		\begin{itemize}
			\item \emph{Variable}: $\trm{M} = x$.
			Then $\trm{N} = x$ since $x$ is normal. 
			Moreover, $\trm{M*} = x$ and hence $\trm{<M*>} = x$.
			Therefore, $\trm{N} \rwp \trm{<M*>}$ because $\rwp_\beta$ is reflexive.
			
			\item \emph{Abstraction}: $\trm{M} = \trm{\x.P}$.
			Then $\trm{N} = \trm{\x.<P^\circ>}$ for some labeling $\trm{P^\circ}$ of $\trm{P}$.
			Since $\trm{P} \rwp_\beta \trm{<P^\circ>}$, one has $\trm{<P^\circ>} \rwp_\beta \trm{<P*>}$ and hence $\trm{N} = \trm{\x.<P^\circ>} \rwp_\beta \trm{\x.<P*>} = \trm{<M*>}$.
			
			\item \emph{Application}: $\trm{M} = \trm{PQ}$.
			Let $\trm{M^\circ}$ be a labeling of $\trm{M}$ such that $N = \trm{<M^\circ>}$.
			There are three subcases:
			\begin{enumerate}
				\item $\trm{M} \neq \trm{(\x.R)Q}$. 
				Then $\trm{N} = \trm{<P^\circ><Q^\circ>}$.
				By \ih\ (since $\trm{P} \rwp_\beta \trm{<P^\circ>}$ and $\trm{Q} \rwp_\beta \trm{<Q^\circ>}$), one has $\trm{<P^\circ>} \rwp_\beta \trm{<P*>}$ and $\trm{<Q^\circ>} \rwp_\beta \trm{<Q*>}$.
				Therefore, by \Cref{lemma:application-parallel-beta}, $\trm{N} \rwp_\beta \trm{<P*><Q*>} = \trm{<M*>}$, where the equality holds because $\trm{M} \neq \trm{(\x.R)Q}$.
				
				\item $\trm{M^\circ} = \trm{(\x.R)*Q}$.
				Then $\trm{N} = \trm{<R^\circ>[<Q^\circ>/x]}$.
				By \ih\ (since $\trm{R} \rwp_\beta \trm{<R^\circ>}$ and $\trm{Q} \rwp_\beta \trm{<Q^\circ>}$), one has $\trm{<R^\circ>} \rwp_\beta \trm{<R*>}$ and $\trm{<Q^\circ>} \rwp_\beta \trm{<Q*>}$.
				By \Cref{lemma:substitution}, $\trm{N} = \trm{<R^\circ>[<Q^\circ>/x]} \rwp_\beta \trm{<R*>[<Q*>/x]}  = \trm{<M*>}$ where the last equality holds because $\trm{M} = \trm{(\x.R)Q}$ and $\trm{M*}$ is the full labeling of $M$.
				
				\item $\trm{M^\circ} \neq \trm{(\x.R)*Q}$ but $\trm{M} = \trm{(\x.R)Q}$ (\ie\ $P = \lambda x.R $). 
				Then $\trm{N} = \trm{<P^\circ><Q^\circ>}$.
				By \ih\ (since $\trm{P} \rwp_\beta \trm{<P^\circ>} = \trm{<\x.R^\circ>}$ and $\trm{Q} \rwp_\beta \trm{<Q^\circ>}$), one has $\trm{<P^\circ>} \rwp_\beta \trm{<P*>} = \trm{<\x.R*>}$ and $\trm{<Q^\circ>} \rwp_\beta \trm{<Q*>}$.
				Therefore, by \Cref{lemma:application-parallel-beta}, $\trm{N} \rwp_\beta \trm{<R*>[<Q*>/x]} = \trm{<M*>}$, where the equality holds because $\trm{M} = \trm{(\x.R)Q}$ and $\trm{M*}$ is the full labeling of $M$.
			\end{enumerate}
		\end{itemize}
	
		\item Since $\trm{M} \rwp \trm{N}$, then $\trm{M} \rwp_\beta \trm{P} \rwn_\perm \trm{N}$ for some term $\trm{P}$ such that $\trm{N} = \trm{P}^\perm$.
		By \Cref{lemma:full-development}.\ref{lemma:full-development-beta}, $\trm{P} \rwp_\beta \trm{<M*>}$ and hence $\trm{P} \rwp \trm{<M*>^\perm}$.
		According to \Cref{lem:p to complete reduction}, $\trm{N} = \trm{P}^\perm \rwp \trm{(<M*>^\perm)^\perm} = \trm{<M*>^\perm}$.
		\qedhere
	\end{enumerate}
\end{proof}

\begin{cor}[Complete reduction is diamond]
\label{cor:complete diamond}
	If $P\rwpleft N\rwp M$ then $P\rwp Q\rwpleft M$ for some $Q$.
\end{cor}

\begin{proof}
	By \Cref{lemma:full-development}.\ref{lemma:full-development-complete}, $P \rwp
	\! \trm{<N*>} \rwpleft M$ where $\trm{N*}$ is the full labeling of $\trm{N}$.
\end{proof}

\begin{figure}[!t]
\[
\begin{array}{cccc}
\vcenter{\hbox{\begin{tikzpicture}
	\matrix (m) [matrix of math nodes] {
	  \trm{N*} && \trm{M*} \\[20pt] \\ \trm{<N*>} & =_\perm & \trm{<M*>} \\
	};
	\draw[rwb] (m-1-1) --node[above]{$\scriptstyle\perm$} (m-1-3);
	\draw[rwp] (m-1-1) --node[left] {$\scriptstyle\beta$} (m-3-1);
	\draw[rwp] (m-1-3) --node[right]{$\scriptstyle\beta$} (m-3-3);
\end{tikzpicture}}}
&
\vcenter{\hbox{\begin{tikzpicture}
	\matrix (m) [matrix of math nodes] {
	  \trm{N*} && \trm{M*} \\[20pt] \\ \trm{<N*>} & =_\perm & \trm{<M*>} \\
	};
	\draw[rwn] (m-1-1) --node[above]{$\scriptstyle\perm$} (m-1-3);
	\draw[rwp] (m-1-1) --node[left] {$\scriptstyle\beta$} (m-3-1);
	\draw[rwp] (m-1-3) --node[right]{$\scriptstyle\beta$} (m-3-3);
\end{tikzpicture}}}
&
\vcenter{\hbox{\begin{tikzpicture}
	\matrix (m) [matrix of math nodes] {
	  \trm{N} &[20pt] & \trm{M} \\[20pt] \\ N^\perm && M^\perm \\
	};
	\draw[rw]  (m-1-1) -- (m-1-3);
	\draw[rwn] (m-1-1) --node[left] {$\scriptstyle\perm$} (m-3-1);
	\draw[rwn] (m-1-3) --node[right]{$\scriptstyle\perm$} (m-3-3);
	\draw[rwx] (m-3-1) -- (m-3-3);
\end{tikzpicture}}}
&
\vcenter{\hbox{\begin{tikzpicture}
	\matrix (m) [matrix of math nodes] {
	  \trm{N} &[20pt] & \trm{M} \\[20pt] \\ \trm{P} && \trm{Q} \\
	};
	\draw[rwp] (m-1-1) -- (m-1-3);
	\draw[rwp] (m-1-1) -- (m-3-1);
	\draw[rwp] (m-1-3) -- (m-3-3);
	\draw[rwp] (m-3-1) -- (m-3-3);
\end{tikzpicture}}}
\\ \\
  \text{\Cref{lem:parallel p - parallel beta}}
& \text{\Cref{lem:exhaustive p - parallel beta}}
& \text{\Cref{lem:p to complete reduction}}
& \text{\Cref{cor:complete diamond}}
\end{array}
\]
\caption{Diagrams for the lemmata leading up to confluence}
\label{fig:confluence diagrams}
\end{figure}

Since complete reduction $\rwp$ is diamond, the confluence of $\rw$ easily follows.

\begin{thm}
\label{thm:confluence}
Reduction $\rw$ is confluent.
\end{thm}

\begin{proof}
By the following diagram. For the top and left areas, by Lemma~\ref{lem:p to complete reduction} any reduction path $N\rws M$ maps onto one $N^\perm \rwx M^\perm$. The main square follows by the diamond property of complete reduction, \Cref{cor:complete diamond}.

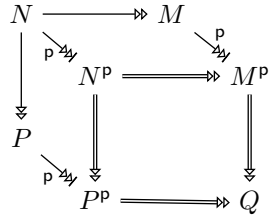
\[
\begin{tikzpicture}
	\matrix (m) [matrix of math nodes] {
	  N &[10pt] &[10pt] M &[10pt] \\[10pt] & N^\perm && M^\perm \\[10pt] P \\[10pt] & P^\perm && Q \\
	};
	\draw[rws] (m-1-1) -- (m-1-3);
	\draw[rws] (m-1-1) -- (m-3-1);
	\draw[rwn] (m-1-1) --node[below left =-2pt] {$\scriptstyle\perm$} (m-2-2);
	\draw[rwn] (m-1-3) --node[above right=-2pt]{$\scriptstyle\perm$} (m-2-4);
	\draw[rwn] (m-3-1) --node[below left =-2pt] {$\scriptstyle\perm$} (m-4-2);
	\draw[rwx] (m-2-2) -- (m-2-4);
	\draw[rwx] (m-2-2) -- (m-4-2);
	\draw[rwx] (m-2-4) -- (m-4-4);
	\draw[rwx] (m-4-2) -- (m-4-4);
\end{tikzpicture}
\qedhere
\]
\end{proof}


\section{Strong normalization for simply typed terms}
\label{sec:SN}

In this section, we prove that the relation $\rw$ enjoys strong
normalization in simply typed terms. Our proof of strong
normalisation is based on the classic reducibility technique,
and inherently has to deal with label-open terms. It thus make great
sense to turn the order $<_{{\trm{M}}}$ from
Definition~\ref{def:orderlabels} into something more formal, at the same time
allowing terms to be label-\emph{open}. This is in
Figure~\ref{fig:order}.
\begin{figure}
  \fbox{
    \begin{minipage}{.97\textwidth}
      \begin{tabular}{ll}        
        \textsf{Label Sequences}: & $\theta\quad\coloneqq\quad \varepsilon ~\mid~ a\cdot\theta$\\[5pt]
        \textsf{Label Judgments}: & $\xi\quad\coloneqq\quad \labjudg{\theta}{M}$\\
        \begin{minipage}{.25\textwidth}\textsf{Label Rules}:\end{minipage} &
        \begin{minipage}{.7\textwidth}
        \[
        \begin{array}{c}        
        \infer{\labjudg{\theta}{x}}{}\qquad\quad
        \infer{\labjudg{\theta}{\x.M}}{\labjudg{\theta}{M}}\qquad\quad
        \infer{\labjudg{\theta}{\ttrm{!a.M}}}{\labjudg{a\cdot\theta}{M}}
        \\\\
        \infer{\labjudg{\theta}{MN}}{\labjudg{\theta}{M} & \labjudg{\theta}{N}}\qquad\quad
        \infer{\labjudg{\theta}{\ttrm{M +a N}}}{\labjudg{\theta}{M} & \labjudg{\theta}{N} & a\in\theta}
        \end{array}
        \]
        \end{minipage}
        \\
      \end{tabular}
    \end{minipage}
  }
  \label{fig:order}
  \caption{Labelling Terms}
\end{figure}
%
It is easy to realize that, of course modulo label $\alpha$-equivalence, for
every term $M$ there is at least one $\theta$ such that $\labjudg{\theta}{M}$.
An easy fact to check is that if $\labjudg{\theta}{M}$ and $M\rw^\theta N$, then $\labjudg{\theta}{N}$.
It thus makes sense to parametrize $\rw$ on a sequence
of labels $\theta$, i.e., one can define a family of reduction
relations $\rw^\theta$ on pairs in the form $(M,\theta)$.
The set of strongly normalizable terms, and the number of steps
to normal forms become themselves parametric:
\begin{itemize}
\item
  The set $\mathit{SN}^\theta$ of those terms $M$ such
  that $\labjudg{\theta}{M}$ and $(M,\theta)$ is strongly
  normalizing modullo $\rw^\theta$;
\item
  The function $\mathit{sn}^\theta$ assigning to any
  term in $\mathit{SN}^\theta$ the maximal number of $\rw^\theta$
  steps to normal form.
\end{itemize}

Now, let us define types, environments, judgments, and typing rules in
Figure~\ref{fig:typing}.
\newcommand{\arrow}{\Rightarrow}
\newcommand{\judg}[3]{#1\vdash #2:#3}
\begin{figure}
  \fbox{
    \begin{minipage}{.97\textwidth}
      \begin{tabular}{ll}        
        \textsf{Types}: & $\tau\quad\coloneqq\quad \alpha ~\mid~ \tau\arrow\rho$\\[5pt]
        \textsf{Environments}: & $\Gamma\quad\coloneqq\quad x_1:\tau_1,\ldots,x_n:\tau_n$\\[5pt]
        \textsf{Judgments}: & $\pi\quad\coloneqq\quad \judg{\Gamma}{M}{\tau}$\\
        \begin{minipage}{.25\textwidth}\textsf{Typing Rules}:\end{minipage} &
        \begin{minipage}{.7\textwidth}
        \[
        \begin{array}{c}
        \infer{\judg{\Gamma,x:\tau}{x}{\tau}}{}\qquad\quad
        \infer{\judg{\Gamma}{\x.M}{\tau\arrow\rho}}{\judg{\Gamma,x:\tau}M\rho}\qquad\quad
        \infer{\judg{\Gamma}{\ttrm{!a.M}}{\tau}}{\judg{\Gamma}{M}{\tau}}
        \\\\
        \infer{\judg{\Gamma}{MN}{\rho}}{\judg{\Gamma}{M}{\tau\arrow\rho} & \judg{\Gamma}{N}{\tau}}\qquad\quad
        \infer{\judg{\Gamma}{\ttrm{M +a N}}{\tau}}{\judg{\Gamma}{M}{\tau} & \judg{\Gamma}{N}{\tau}}
        \end{array}
        \]
        \end{minipage}
        \\
      \end{tabular}
    \end{minipage}
  }
  \caption{Types, Environments, Judgments, and Rules}
  \label{fig:typing}
\end{figure}
Please notice that the type structure is precisely the one of the
usual, vanilla, simply-typed $\lambda$-calculus (although terms are of
course different),

\begin{lem}
\label{lemma:cloredsum}
  The following rule is sound:
  $$
  \infer{\trm{M +a N}L_1\in\mathit{SN}^\theta\ldots L_m\in\mathit{SN}^\theta}{ML_1\ldots L_m & NL_1\ldots L_m\in\mathit{SN}^\theta & a\in\theta}
  $$
\end{lem}
\begin{proof}
  Let $h(a,\theta)$ be the height of $a$ in the sequence $\theta$.
  The proof goes by lexicographic induction on the following quadruple:
  \begin{equation}\label{equ:redord}
  (m,h(a,\theta),\sum_{i=1}^m\mathit{sn}^\theta(L_i)+\mathit{sn}^\theta(M)+\mathit{sn}^\theta(N),|M|+|N|)
  \end{equation}
  We proceed by showing that all terms to which
  $\trm{M +a N}L_1\ldots L_m$ reduces are in $\mathit{SN}^\theta$, and thus
  $\trm{M +a N}L_1\ldots L_m$ is itself in $\mathit{SN}^\theta$:
  \begin{itemize}
  \item
    If reduction happens in one between $M,N,L_1,\ldots,L_m$, then we
    are done, because we get an instance of the rule where the first two
    components of (\ref{equ:redord}) stay constant, and the third
    strictly decreases.
  \item
    If reduction happens at the leftmost component $\trm{M +a N}$
    due to the rule $\trm{P +a P}\rw_\perm P$, then one of the two
    hypotheses apply trivially.
  \item
    If reduction happens at the leftmost component $\trm{M +a N}$
    due to the rule $\trm{(P +a Q) +a R}\rw_\perm \trm{P +a R}$,
    then the \ih\ holds due to the fourth component being strictly
    smaller. Similarly when the rule is
    $\trm{P +a (Q +a R)}\rw_\perm \trm{P +a R}$.    
  \item
    If reduction happens in the leftmost application due to
    the rule $\trm{(P +a Q) R}\rw_\perm \trm{(PR) +a (QR)}$,
    then the \ih\ can be applied, because the number of arguments $m$
    strictly decreases. Similarly when the rule is
    $\trm{N (M +a P)}\rw_\perm\trm{(NM) +a (NP)}$.
  \item
    If reduction happens at the leftmost component $\trm{M +a N}$
    due to the rule $\trm{(P +b Q) +a R}\rw_\perm \trm{(P +a R) +b (Q +a R)}$,
    then $M$ can be written as $\trm{P +b Q}$, and the
    following two terms are strongly normalizing (because
    $\trm{(P +b Q)L_1\ldots L_m}$ is by hypothesis strongly normalizing
      itself):
    $$
    PL_1\ldots L_m\qquad QL_1\ldots L_m
    $$
    Moreover, $\mathit{sn}^\theta(PL_1\ldots L_m),\mathit{sn}^\theta(QL_1\ldots
    L_m)\leq\mathit{sn}^\theta(\trm{(P +b Q)L_1\ldots L_m}$.  We can then
    conclude that $\trm{(P +a N)}L_1\ldots L_m$ and $\trm{(Q +a
        N)}L_1\ldots L_m$ are both strongly normalizing by \ih, because
    the fourth component is strictly smaller. Finally, we can conclude
    again by induction hypothesis, since the number of arguments stays
    the same, while the level of $b$ must be smaller than that of
    $a$. Similarly if the rule applied is $\trm{P +a (Q +b
      R)}\rw_\perm \trm{(P +a Q) +b (P +a R)}$.
    \qedhere
  \end{itemize}
\end{proof}

\newcounter{lemma:cloredbox}
\addtocounter{lemma:cloredbox}{\value{defn}}
\begin{lem}\label{lemma:cloredbox}
  The following rule is sound:
  \marginpar{\footnotesize Proof in the Appendix}
  $$
  \infer{\trm{(!a.M)}L_1\ldots L_m\in\mathit{SN}^\theta}{\trm{M}L_1\ldots L_m\in\mathit{SN}^{a\cdot\theta} & \forall i.a\not\in L_i}
  $$
\end{lem}

\begin{lem}\label{lemma:cloredheadvar}
  The following rule is sound
  $$
  \infer{xL_1\ldots L_m\in\mathit{SN}^\theta}{L_1\in\mathit{SN}^\theta &\cdots & L_m\in\mathit{SN}^\theta}
  $$
\end{lem}

\begin{proof}
	Trivial, since the term $xL_1\ldots L_m$ cannot create new redexes.
\end{proof}

\newcounter{lemma:cloredbeta}
\addtocounter{lemma:cloredbeta}{\value{defn}}
\begin{lem}\label{lemma:cloredbeta}
  The following rule is sound
    \marginpar{\footnotesize Proof in the Appendix}
  $$
  \infer{\trm{(\x.M)}L_0\ldots L_m\in\mathit{SN}^\theta}{\trm{M[L_0/x]}L_1\ldots L_m\in\mathit{SN}^\theta & L_0\in\mathit{SN}^\theta}
  $$
\end{lem}

\newcommand{\RedSet}[1]{\mathit{Red}_{#1}}
The definition of a reducible term stays the same:
\begin{align*}
  \RedSet{\alpha}&=\{(\Gamma,\theta,M)\mid M\in\mathit{SN}^\theta\wedge\judg{\Gamma}{M}{\alpha}\};\\
  \RedSet{\tau\arrow\rho}&=\{(\Gamma,\theta,M)\mid(\judg{\Gamma}{M}{\tau\arrow\rho})\wedge(\labjudg{\theta}{M})\wedge\\
    &\qquad\forall(\Gamma\Delta,\theta,N)\in\RedSet{\tau}.(\Gamma\Delta,\theta,MN)\in\RedSet{\rho}\}.
\end{align*}
\begin{lem}\label{lemma:redprop}
  \begin{enumerate}
  \item\label{point:impl}
    If $(\Gamma,\theta,M)\in\RedSet{\tau}$, then $M\in\mathit{SN}^\theta$;
  \item\label{point:headvar}
    If $\judg{\Gamma}{x L_1\ldots L_m}{\tau}$ and $L_1,\ldots,L_m\in\mathit{SN}^\theta$,
    then $(\Gamma,\theta,x L_1\ldots L_m)\in\RedSet{\tau}$.
  \item\label{point:beta}
    If $(\Gamma,\theta,\trm{M[L_0/x]}L_1\ldots L_m)\in\RedSet{\tau}$,
    $\judg{\Gamma}{L_0}{\rho}$ and $L_0\in\mathit{SN}^\theta$
    then $(\Gamma,\theta,\trm{(\x.M)}L_0\ldots L_m)\in\RedSet{\tau}$.
  \item\label{point:sum}
    If $(\Gamma,\theta,ML_1\ldots L_m)\in\RedSet{\tau}$,
    $(\Gamma,\theta,NL_1\ldots L_m)\in\RedSet{\tau}$, and $a\in\theta$,
    then $(\Gamma,\theta,\trm{(M +a N)}L_1\ldots L_m)\in\RedSet{\tau}$.
  \item\label{point:box}
    If $(\Gamma,a\cdot\theta,ML_1\ldots L_m)\in\RedSet{\tau}$
    and $a\not\in L_i$ for every $i$,
    then $(\Gamma,\theta,\trm{(!a.M)}L_1\ldots L_m)\in\RedSet{\tau}$.
  \end{enumerate}
\end{lem}
\begin{proof}
  The proof is an induction on $\tau$:
  \begin{itemize}
  \item
    If $\tau$ is an atom $\alpha$, then Point \ref{point:impl} follows
    by definition, while points \ref{point:headvar} to \ref{point:box}
    come from Lemma~\ref{lemma:cloredsum}, Lemma~\ref{lemma:cloredbox},
    Lemma~\ref{lemma:cloredheadvar}, and Lemma~\ref{lemma:cloredbeta}.
  \item
    If $\tau$ is $\rho\arrow\mu$, points \ref{point:headvar} to \ref{point:box}
    come directly from the induction hypothesis, while \ref{point:impl}
    can be proved by observing that $M$ is in $\mathit{SN}^\theta$
    if $Mx$ is itself $\mathit{SN}^\theta$, where $x$ is a fresh variable.
    By induction hypothesis (on point~\ref{point:headvar}), we can
    say that $(\Gamma(x:\rho),\theta,x)\in\RedSet{\rho}$, and
    conclude that $(\Gamma(x:\rho),\theta,Mx)\in\RedSet{\mu}$.
    \qedhere
  \end{itemize}
\end{proof}

The following is the so-called Main Lemma:
\begin{prop}\label{lemma:mainlemma}
  Suppose $\judg{y_1:\tau_1,\ldots,y_n:\tau_n}{M}{\rho}$,
  $\labjudg{\theta}{M}$, and let
  $(\Gamma,\theta,N_j)\in\RedSet{\tau_j}$ for all $1\leq j\leq
  n$. Then
  $(\Gamma,\theta,M[N_1/y_1,\ldots,N_n/y_n])\in\RedSet{\rho}$.
\end{prop}
\begin{proof}
  This is an induction on the structure of the term $M$:
  \begin{itemize}
  \item
    If $M$ is a variable, necessarily one among $y_1,\ldots,y_n$,
    then the result is trivial.
  \item
    If $M$ is an application $LP$, then there exists
    a type $\xi$ such that
    $\judg{y_1:\tau_1,\ldots,y_n:\tau_n}{L}{\xi\arrow\rho}$
    and $\judg{y_1:\tau_1,\ldots,y_n:\tau_n}{P}{\xi}$. Moreover,
    $\labjudg{\theta}{L}$ and $\labjudg{\theta}{P}$ we can
    then safely apply the induction hypothesis and conclude
    that
    $$
    (\Gamma,\theta,L[\overline{N}/\overline{y}])\in\RedSet{\xi\arrow\rho}
    \qquad
    (\Gamma,\theta,P[\overline{N}/\overline{y}])\in\RedSet{\xi}
    $$
    By definition, we get
    $$
    (\Gamma,\theta,(LP)[\overline{N}/\overline{y}])\in\RedSet{\rho}
    $$
  \item
    If $M$ is an abstraction $\trm{\x.L}$, then $\rho$
    is an arrow type $\xi\arrow\mu$ and
    $\judg{y_1:\tau_1,\ldots,y_n:\tau_n,x:\xi}{L}{\mu}$.
    Now, consider any $(\Gamma\Delta,\theta,P)\in\RedSet{\xi}$.
    Our objective is to prove with this hypothesis
    that $(\Gamma\Delta,\theta,(\x.L[\overline{N}/\overline{y}])P)\in\RedSet{\mu}$
    By induction hypothesis, since $(\Gamma\Delta,N_i)\in\RedSet{\tau_i}$,
    we get that $(\Gamma\Delta,\theta,L[\overline{N}/\overline{y},P/x])\in\RedSet{\mu}$.
    The thesis follows from Lemma~\ref{lemma:redprop}.
  \item
    If $M$ is a sum $\trm{L+a P}$, we can make use of Lemma~\ref{lemma:redprop}
    and the induction hypothesis, and conclude.
  \item
    If $M$ is a sum $\trm{!a.P}$, again, we can make use of Lemma~\ref{lemma:redprop}
    and the induction hypothesis. We should however observe, that $\labjudg{a\cdot\theta}{P}$,
    since $\labjudg{\theta}{M}$.
    \qedhere
  \end{itemize}
\end{proof}
We now have all the 
ingredients for our proof of strong normalization:
\begin{theorem}
  If $\judg{\Gamma}{M}{\tau}$ and $\labjudg{\theta}{M}$, then $M\in\mathit{SN}^\theta$.
\end{theorem}
\begin{proof}
  Suppose that $\judg{x_1:\rho_1,\ldots,x_n:\rho_n}{M}{\tau}$.
  Since $\judg{x_1:\rho_1,\ldots,x_n:\rho_n}{x_i}{\rho_i}$ for all
  $i$, and clearly $\labjudg{\theta}{x_i}$ for every $i$, we
  can apply Lemma \ref{lemma:mainlemma} and
  obtain that $(\Gamma,\theta,M[\overline{y}/\overline{y}])\in\RedSet{\tau}$
  from which, via Lemma~\ref{lemma:redprop}, one gets the thesis.
\end{proof}


\section{Projective reduction}
\label{sec:projective-reduction}

Permutative reduction $\rw_\perm$ evaluates probabilistic sums purely by rewriting. Here we look at an alternative \emph{projective} notion of reduction, which conforms more closely to the intuition that $\ttrm{!a}$ generates a probabilistic event to determine the choice $\ttrm{+a}$. We would like to evaluate $\ttrm{!a.N}$ by the probabilistic sum $N_0+N_1$, where $N_i$ is $N$ with any subterm of the form $\ttrm{M_0+aM_1}$ projected to $M_i$. The meta-level sum $+$ must be distinct from the abbreviation $N\+M=\ttrm{!a.N+aM}$, since otherwise reduction on this term would be circular. We then need a call-by-need strategy for evaluating $\ttrm{!a.N}$, which we implement through a \emph{head context}.

\begin{defn}
The \emph{$a$-projections} $\proj a0N$ and $\proj a1N$ are defined as follows, where $a\neq b$.
\begin{align*}
	\proj a0{N+aM} &= \proj a0N		&	\proj ai{\x.N} &= \x.\proj aiN
\\	\proj a1{N+aM} &= \proj a1M		&	\proj ai{NM}   &= (\proj aiN)(\proj aiM)
\\	\proj ai{!a.N} &= \trm{!a.N}	& 	\proj ai{N+bM} &= (\proj aiN)\trm{+b}(\proj aiM)
\\	\proj aix      &= x				&	\proj ai{!b.N} &= \trm{!b.}\proj aiN
\end{align*}
\end{defn}

\begin{defn}
A \emph{head context} $H[\,]$ is given by the following grammar.
\[
	H[\,] \coloneqq [\,] ~\mid~ \lambda x.H[\,] ~\mid~ H[\,]N
\]
\end{defn}

\begin{defn}
\emph{Projective} reduction $\rw_\pi$ is the following reduction step.
\[
	H[\,\trm{!a.N}] ~\rw_\pi~ H[\proj a0N] + H[\proj a1N]
\]
\end{defn}

Projective reduction is a strategy, and not a rewrite relation, as it applies only in head contexts---that is, we do not evaluate inside arguments or under other generators $\ttrm{!a}$ or choices $\ttrm{+a}$. This corresponds to the generator permuting outside of a head context by $\boxAbs$ and $\boxFun$ steps, as below.
\[
	\ttrm{H[\,!a.N]}\rws_\perm \ttrm{!a.H[N]}
\]
Evaluating inside $\ttrm{!a}$ would correspond to permuting two generators, $\ttrm{!a.!b.N}$ to $\ttrm{!b.!a.N}$, which we do not admit as a rewrite rule. Similarly, we do not evaluate inside a choice, since we distribute a generator over a choice (rule $\plusBox$) rather than permuting in the other direction. Still, since projective reduction evaluates the outermost generator on the spine of a term, and removes all its bound labels, projective reduction has the same normal forms except in argument position. Rather than trying to bring both forms closer together, which would involve juggling $\+$ and $+$ and several kinds of reduction context, we contend ourselves with showing that permutative reduction can implement a projective step. For the following proposition, observe that if $a$ is not free in $N$, then $\proj aiN=N$, so that we have $H[\,\trm{!a.N}] ~\rw_\pi~ H[N]+H[N]~=~H[N]$. 



\begin{prop}[Permutative reduction simulates projective reduction]
\[
	H[\,\trm{!a.N}] ~\rws_\perm~
	\left\{\begin{array}{l@{\qquad}l}
		H[N]						 & (a\notin\fl N) \\[5pt]
		H[\proj a0N] \+ H[\proj a1N] & (a\in\fl N)
	\end{array}\right.
\]
\end{prop}

\begin{proof}
The case $a\notin\fl N$ is immediate by a $\boxVoid$ step. For the other case, observe that $\trm{H[\,!a.N]}\rws_\perm\trm{!a.H[N]}$ by $\boxAbs$ and $\boxFun$ steps, and since $a$ does not occur in $H[\,]$, we have $H[\proj aiN]=\proj ai{H[N]}$. Unfolding also the definition of $\+$, we may thus restrict ourselves to proving the following.
\[
	\trm{!a.N} ~\rws_\perm~ \trm{!a.} \proj a0N~\trm{+a}~\proj a1N
\]
Given that $a$ is the smallest label in $N$, \ie\ $a<b$ for any other label $b$ in $N$, a straightforward induction on $N$ gives $N\rws_\perm\proj a0N~\ttrm{+a}~\proj a1N$, with as special base case $N\rws_\perm N$ for $a\notin\fl N$.
\end{proof}

\newcommand\val{\mathsf{v}}
\newcommand\plusval{\mathbin{\smallbin\oplus_\val}}

\section{Call-by-value translation}
\label{sec:cbv}

We consider the interpretation of a call-by-value probabilistic lambda-calculus. For simplicity we will only restrict probabilistic reduction to a call-by-value regime, and not $\beta$-reduction; our values $V$ are then just deterministic (the idea is that probabilistic sums cannot be duplicated or erased).
\[
\begin{array}{l@{\qquad}rcl}
	N \coloneqq x \mid \x.N \mid MN \mid M\plusval N  & (\x.N)V &\rw_\val& N[V/x]
\\[5pt]
	V \coloneqq x \mid \x.V \mid VW  &  M\plusval N &\rw_\val& M + N
\end{array}
\]
The call-by-value interpretation $\uncbv{N}$ of term $N$ in this calculus is given as follows. First, we translate $N$ to an open $\PEL$-term $\unopen N=\labjudg\theta P$, and then $\uncbv{N}$ is the label closure $\uncbv{N}=\labclose\theta P$, which prefixes $P$ with a generator $\trm{!a}$ for every $a$ in $\theta$.

\begin{defn}
The \emph{label closure} $\labclose\theta P$ is given inductively as follows.
\[
	\labclose{}P = P \qquad \labclose{a\cdot\theta}P = \labclose\theta{\trm{!a.P}}
\]
The \emph{open interpretation} $\unopen N$ is given as follows, where $\unopen{N_i}=\labjudg{\theta_i}{P_i}$ for $i\in\{1,2\}$, and all labels are fresh.
\[
\begin{array}{r@{\quad}c@{\quad}l@{\qquad}r@{\quad}c@{\quad}l}
		\unopen x 		 &=& \labjudg{}x				& \unopen {N_1N_2}   &=& \labjudg{\theta_2\cdot\theta_1}{P_1P_2}
\\[5pt]	\unopen {\x.N_1} &=& \labjudg{\theta_1}{\x.P_1} & \unopen {N_1\plusval N_2} &=& \labjudg{\theta_2\cdot\theta_1\cdot a}{\trm{P_1 +a P_2}}
\end{array}
\]
The \emph{call-by-value interpretation} of a probabilistic lambda-term $N$ is
\[
	\uncbv N=\lfloor\unopen N\rfloor~.
\]
\end{defn}

The choice of ordering in building up $\theta$ in an open interpretation will determine the order of the labels in a term, and thus which label will rise to the surface. If we want to simulate a call-by-value step $M\plusval N\rw_\val M+N$ with $\rws_\perm$, we need to permute $\theta$ so that the label $a$ assigned to the translation of the $\plusval$ will be the smallest.

\begin{prop}
If $N\rw_\val M+P$ and $\unopen N=\labjudg\theta{N'}$ then there is a permutation $\theta'\cdot a$ of $\theta$ that gives the following.
\[
	\labclose{\theta'\cdot a}{N'}~\rws_\perm~\trm{!a.}\uncbv M\trm{+a}\uncbv P
\]
\end{prop}

\begin{proof}
	The result follows from the following, more general, statement:
	if $N \rw_\val M +P$ and $\unopen N=\labjudg\theta{N'}$ then there is a permutation  $\theta'' \cdot \theta' \cdot a$ of $\theta$ such that  
	$\labclose{\theta'' \cdot \theta' \cdot a}{N'} \rws_\perm \labclose{\theta'' \cdot \theta' \cdot a}{Q \trm{+a} R}$ where $\unopen{M} =  \labjudg{\theta' }{Q}$ and  $\unopen{P} = \labjudg{\theta''}{R}$.
	The proof of this statement is by induction on the structure of the context $C$ such that $N=C[Q\oplus_\val R]$ and $M=C[Q]$ and $P=C[R]$.
\end{proof}

\bibliographystyle{splncs04}
\bibliography{biblio}
\addcontentsline{toc}{section}{References}

\newpage
\appendix

\section{Technical appendix: omitted proofs}
\label{sect:proofs}

The enumeration of lemmas already stated in the body of the article is unchanged.

\setcounter{lemmaAppendix}{\value{lem:confluence-perm}}
\begin{lemmaAppendix}[Confluence of $\rw_\perm$]
	\label{lemmaAppendix:confluence-perm}
	Reduction $\rw_\perm$ is confluent.
\end{lemmaAppendix}

\begin{proof}
%
\newcommand\itm[2]{\medskip\noindent(#1)}
	We consider only the cases omitted in the proof on p.~\pageref{lem:confluence-perm}.
	
	\itm\plusFun{\trm{(N+aM)P}\rw\trm{(NP)+a(MP)}}
	{\small
		\[
		\vcenter{\hbox{\begin{tikzpicture}
				\matrix [matrix of math nodes] (m) {
					\trm{(N+aM)(P+aQ)}                &[10pt] \trm{(N(P+aQ))+a(M(P+aQ))}
					\\[20pt]\trm{((N+aM)P)+a((N+aM)Q)}        &       \trm{((NP)+a(NQ))+a((MP)+a(MQ))}
					\\[20pt]\trm{((NP)+a(MP))+a((NQ)+a(MQ))}  &       \trm{(NP)+a(MQ)}
					\\ };
				\draw[rw] (m-1-1) --node[above]{$\plusFun$} (m-1-2);
				\draw[rw] (m-1-1) --node[left] {$\plusArg$} (m-2-1);
				\draw[rws,implied] (m-1-2) --node[right]{$\plusArg$} (m-2-2);
				\draw[rws,implied] (m-2-1) --node[left] {$\plusFun$} (m-3-1);
				\draw[rws,implied] (m-2-2) --node[right]{$\cancelL,\cancelR$} (m-3-2);
				\draw[rws,implied] (m-3-1) --node[below]{$\cancelL,\cancelR$} (m-3-2);
				\end{tikzpicture}}}
		\]
		\[
		\vcenter{\hbox{\begin{tikzpicture}
				\matrix [matrix of math nodes] (m) {
					\trm{(N+bM)(P+aQ)}         &[20pt] \trm{(N(P+aQ))+b(M(P+aQ))}
					\\[20pt]                           &       \trm{((NP)+a(NQ))+b((MP)+a(MQ))} 
					\\[20pt]\trm{((N+bM)P)+a((N+bM)Q)} &       \trm{((NP)+b(MP))+a((NQ)+b(MQ))}
					\\ };
				\draw[rw] (m-1-1) --node[above]{$\plusFun$} (m-1-2);
				\draw[rw] (m-1-1) --node[left] {$\plusArg$} (m-3-1);
				\draw[rws,implied] (m-1-2) --node[right]{$\plusArg$} (m-2-2);
				\draw[rws,implied] (m-2-2) --node[right]{$\plusL,\plusR,\cancelL,\cancelR$} (m-3-2);
				\draw[rws,implied] (m-3-1) --node[below]{$\plusFun$} (m-3-2);
				\end{tikzpicture}}}
		\]
	}
	
	\itm\plusBox{\trm{!b.(N +a M)} \rw \trm{(!b.N) +a (!b.M)}}
	{\small
		\[
		\vcenter{\hbox{\begin{tikzpicture}
				\matrix [matrix of math nodes] (m) {
					\trm{!b.(N+aM)}    &[20pt] \trm{(!b.N)+a(!b.M)}
					\\[20pt]\trm{N+aM}
					\\ };
				\draw[rw] (m-1-1) --node[above]{$\plusBox$} (m-1-2);
				\draw[rw] (m-1-1) --node[left] {$\boxVoid$} (m-2-1);
				\draw[rws,implied] (m-1-2) --node[below right=-2pt] {$\boxVoid$} (m-2-1);
				\end{tikzpicture}}}
		\quad (a \neq b, \ b \notin\trm{N+a M})
		\]
	}
\end{proof}

\setcounter{lemmaAppendix}{\value{lemma:application-parallel-beta}}
\begin{lemmaAppendix}
	\label{lemmaAppendix:application-parallel-beta}
	If $\trm{M} \rwp_\beta \trm{M'}$ and $\trm{N} \rwp_\beta \trm{N'}$, 
	then $\trm{MN} \rwp_\beta \trm{M'N'}$.
	If moreover $M = \lambda x.R$ and $M'^ = \lambda x R'$ with $R \rwp_\beta R'$, then  $\trm{MN} \rwp_\beta \trm{R'[N'/x]}$.
\end{lemmaAppendix}

\begin{proof}
	Since $\trm{M} \rwp_\beta \trm{<M*>} = M'$ and $\trm{N} \rwp_\beta \trm{<N*>} = N'$ for some labelings $\trm{M*}$ and $\trm{N*}$ of $M$ and $N$ respectively, let $\trm{(MN)*}$ be the labeling of $\trm{MN}$ obtained by putting together the labelings $\trm{M*}$ and $\trm{N*}$.
	Clearly, $\trm{MN} \rwp_\beta \trm{<(MN)*>} = \trm{<M*><N*>} = M'N'$.
	
	Concerning the part of the statement after ``moreover'', since $\trm{R} \rwp_\beta \trm{<R*>} = R'$ and $\trm{N} \rwp_\beta \trm{<N*>} = N'$ for some labelings $\trm{R*}$ and $\trm{N*}$ of $R$ and $N$ respectively, let $\trm{(MN)*}$ be the labeling of $\trm{MN}$ obtained by putting together the labelings $\trm{R*}$ and $\trm{N*}$ plus the labeled $\beta$-redex $\trm{(\x.R)*N}$.
	Clearly, $\trm{MN} \rwp_\beta \trm{<(MN)*>} = \trm{<R*>[<N*>/x]} = \trm{R'[N'/x]}$.
\end{proof}

\setcounter{lemmaAppendix}{\value{lemma:cloredbox}}
\begin{lemmaAppendix}\label{lemmaAppendix:cloredbox}
	The following rule is sound:
	$$
	\infer{\trm{(!a.M)}L_1\ldots L_m\in\mathit{SN}^\theta}{\trm{M}L_1\ldots L_m\in\mathit{SN}^{a\cdot\theta} & \forall i.a\not\in L_i}
	$$
\end{lemmaAppendix}
\begin{proof}
	The proof is structurally very similar to the one of
	Lemma~\ref{lemma:cloredsum}. The lexicographic order
	is however the following one:
	$$
	(m,\sum_{i=1}^m\mathit{sn}^{a\cdot\theta}(L_i)+\mathit{sn}^{a\cdot\theta}(M),|M|)
	$$
	There is one case in which we need to use
	Lemma \ref{lemma:cloredsum}, namely the one in which the
	considered reduction rule is the following:
	$$
	\trm{!b.(P +a Q)}\rw_\perm\trm{(!b.P) +a (!b.Q)}.
	$$
	Since $M=\trm{P +a Q}$ and $ML_1\ldots L_m$ by hypothesis,
	we conclusde that $PL_1\ldots L_m$ and $QL_1\ldots L_m$
	are both strongly normalizing. By induction hypothesis,
	since $\mathit{sn}^{a\cdot\theta}(P),\mathit{sn}^{a\cdot\theta}(Q)\leq\mathit{sn}^{a\cdot\theta}(M)$,
	it holds that $\trm{(!b.P)}L_1\ldots L_m$ and $\trm{(!b.Q)}L_1\ldots L_m$
	are both strongly normalizing themselves. Lemma~\ref{lemma:cloredsum},
	yields the thesis.
\end{proof}

%

\setcounter{lemmaAppendix}{\value{lemma:cloredbeta}}
\begin{lemmaAppendix}\label{lemmaAppendix:cloredbeta}
	The following rule is sound
	$$
	\infer{\trm{(\x.M)}L_0\ldots L_m\in\mathit{SN}^\theta}{\trm{M[L_0/x]}L_1\ldots L_m\in\mathit{SN}^\theta & L_0\in\mathit{SN}^\theta}
	$$
\end{lemmaAppendix}
\begin{proof}
	Again, the proof is structurally very similar to the one
	of Lemma~\ref{lemma:cloredsum}. The underlying order,
	needs to be slightly adapted, and is the lexicographic
	order on
	\begin{equation}\label{equ:redordbet}
	(\mathit{sn}(\trm{M[L_0/x]}L_1\ldots L_m)+\mathit{sn}(L_0),|M|)
	\end{equation}
	As usual, we proceed by showing that all terms to which
	$\trm{(\x.M)}L_0\ldots L_m$ reduces are strongly normalizing:
	\begin{itemize}
		\item
		If reduction happens in $L_0$, then we can mimick
		the same reduction in by zero or more reduction
		steps in $\trm{M[L_0/x]}L_1\ldots L_m$, and conclude
		by induction hypothesis, because the first component
		of (\ref{equ:redordbet}) strictly decreases.
		\item
		If reduction happens in $M$ or in $L_1,\ldots,L_m$,
		then we can mimick the same reduction in one or more
		reduction in $\trm{M[L_0/x]}L_1\ldots L_m$, and conclude
		by induction hypothesis since, again, the first component
		of (\ref{equ:redordbet}) strictly decreases.
		\item
		If the reduction step we perform reduces
		$\trm{(\x.M)}L_0$, then the thesis follows from the
		hypothesis about $\trm{M[L_0/x]}L_1\ldots L_m$.
		\item
		If $M$ is in the form $\trm{P +a Q}$ and
		the reduction step we perform reduces
		$\trm{(\x.M)}L_0\ldots L_m$
		to $\trm{((\x.P)+a(\x.Q))}L_0\ldots L_m$,
		we proceed by observing that
		$\trm{(\x.P)}L_0\ldots L_m$
		and $\trm{(\x.Q)}L_0\ldots L_m$ are
		both in $\mathit{SN}^\theta$ and we
		can apply the induction hypothesis to them,
		because the first component of (\ref{equ:redordbet})
		stays the same, but the second one strictly decreases.
		We then obtain that
		$\trm{P[x/L_0])}L_1\ldots L_m$
		and $\trm{Q[x/L_0]}L_1\ldots L_m$ are both
		in $\mathit{SN}^\theta$, and from Lemma~\ref{lemma:cloredsum}
		we get the thesis.
		\item
		If $M$ is in the form $\trm{!a.P}$ and
		the reduction step we perform reduces
		$\trm{(\x.M)}L_0\ldots L_m$ to
		$\trm{!a.(\x.P)}L_0\ldots L_m$. We can
		first of all that we can assume that
		$a\not\in L_i$ for every $0\leq i\leq m$.
		We proceed by observing that
		$\trm{(\x.P)}L_0\ldots L_m$ is in $\mathit{SN}^\theta$
		and we can apply the \ih\ to it, because
		the first component of (\ref{equ:redordbet})
		stays the same, but the second one strictly decreases.
		We then obtain that
		$\trm{P[x/L_0]}L_1\ldots L_m$
		is in $\mathit{SN}^\theta$, and from Lemma~\ref{lemma:cloredbox}
		we get the thesis.
		\qedhere
	\end{itemize}
\end{proof}

\end{document}